\documentclass[12pt,a4paper]{article}
\usepackage[utf8]{inputenc}
\usepackage[english]{babel}
\usepackage{marginnote}
\usepackage{amsmath}
\usepackage{amsfonts}
\usepackage{amssymb}
\usepackage{graphicx}
\usepackage{hyperref}
\usepackage{todonotes}
\usepackage{amsthm}

\newtheorem{myrule}{Rule}

\newtheorem{theorem}{Theorem}
\newtheorem{lemma}{Lemma}[theorem]

\author{Hans L. Bodlaender\thanks{Department of Information and Computing Sciences, Utrecht University, P.O.Box 80.089, 3508 TB Utrecht,
The Netherlands. 
Email: \href{mailto:h.l.bodlaender@uu.nl}{h.l.bodlaender@uu.nl}.
The research of this author was partially supported by the {\em Networks} project, supported by the Netherlands Organization for Scientific Research N.W.O.} \and
Benjamin Burton%
\thanks{School of Mathematics and Physics, The University of Queensland
Brisbane QLD 4072, Australia.
Email: \href{mailto:bab@maths.uq.edu.au}{bab@maths.uq.edu.au}.
This author was supported by the Australian Research Council under the Discovery Projects scheme
(DP150104108).
} \and
Fedor V. Fomin\thanks{Department of Informatics, University of Bergen, 5020 Bergen, Norway \href{mailto:fedor.fomin@ii.uib.no}{fedor.fomin@ii.uib.no}. Supported by the Research Council of Norway via the project   ``MULTIVAL".}
\and
Alexander Grigoriev\thanks{School of Business and Economics, Maastricht University, Maastricht, The Netherlands \href{mailto:a.grigoriev@maastrichtuniversity.nl}{a.grigoriev@maastrichtuniversity.nl}}}

\title{Knot Diagrams of Treewidth Two%
\thanks{A large part of this research was done during the workshop {\em Fixed Parameter Computational
Geometry} at the Lorentz Center in Leiden.}}

\date{}

\begin{document}

\maketitle

\begin{abstract}
%In this paper, we study knot diagrams for which the underlying graph
%has treewidth two. We give linear time algorithms for the following problems:
%given a knot diagram of treewidth two, does it represent the unknot; for a one-dimensional link
%diagram of treewidth two, does it represent the unlink; given two knot diagrams;
%for two knot diagrams of treewidth two, do they represent equivalent knots? 
In this paper, we study knot diagrams for which the underlying graph has treewidth two. We give a linear time algorithm for the following problem: given a knot diagram of treewidth two, does it represent the unknot? 
We also show that for a link diagram of treewidth two we can test in linear time if it represents the unlink. From the algorithm, it follows that a
diagram of the unknot of treewidth 2 can always be reduced to the trivial diagram with at most $n$ (un)twist and (un)poke Reidemeister moves.

\end{abstract}

\section{Introduction}
\label{section:introduction}
A \emph{knot} is a
piecewise linear closed curve $S^1$ embedded
 into the $3$-sphere $S^3$ (or the three-dimensional Euclidean space $\mathbb{R}^3$).
 Two knots are said to be \emph{equivalent} if there is an ambient isotopy between them.
 In other words, two knots are equivalent if it is possible to distort one knot into the other without breaking it.
 The basic problem of knot theory is the following unknotting problem: given a knot, determine whether it is  equivalent to a knot that  bounds an embedded disk in $S^3$. Such a knot is called \emph{unknot}.

 Despite a significant progress, the computational complexity of the unknotting problem remains open. Even the existence of \emph{any} algorithm for this problem is a highly non-trivial question. As was stated by Turing in 1954 in \cite{turing1954solvable}, ``No systematic method is yet known by which one can tell whether two knots are the same." The first algorithm resolving this problem is due to Haken~\cite{Haken61}.
  By the celebrated result of Hass,   Lagarias,   and Pippenger~\cite{HassLP99},  unknot  recognition is in NP.  The problem is also suspected to be in 
 %\todo[inline]{Agol gave a sketch of the proof that the problem is also in  
co-NP (assuming the Generalized Riemann Hypothesis), see  the work of Kuperberg \cite{Kuperb14}. However, no polynomial algorithm for the unknotting problem is known.
 
 It was understood already in 1920s that the 
 question about equivalence of  knots in $\mathbb{R}^3$ is  reducible to a combinatorial question about knot diagrams~\cite{AlexanderBr26,Reidemester27}.  Knot diagrams  are labeled planar graphs  representing a projection of the knot onto a plane. Thus every vertex of the graph in knot diagram is of degree 4 and edges are marked as overcrossing  and undercrossing, see Section~\ref{section:preliminaries} for a formal definition.
%  
% A \emph{knot diagram} is a piecewise linear projection of a knot onto the plane, where the only multiple points are crossings at which one section of the knot crosses another transversely.
% 
 It is one of the most fundamental theorems in knot theory from 1920s 
that any
two diagrams of a knot or link in $\mathbb{R}^3$ differ by a sequence of Reidemeister moves \cite{Reidemester27}, 
illustrated in Fig.~\ref{figure:reidemeister}.  We refer to these moves as (I) twist moves, (II) poke moves, and (III) slide moves, with the reverse operation of a twist move the untwist, and the reverse operation of a poke the unpoke.

\begin{figure}[htb]
\begin{center}
\includegraphics[scale=0.8]{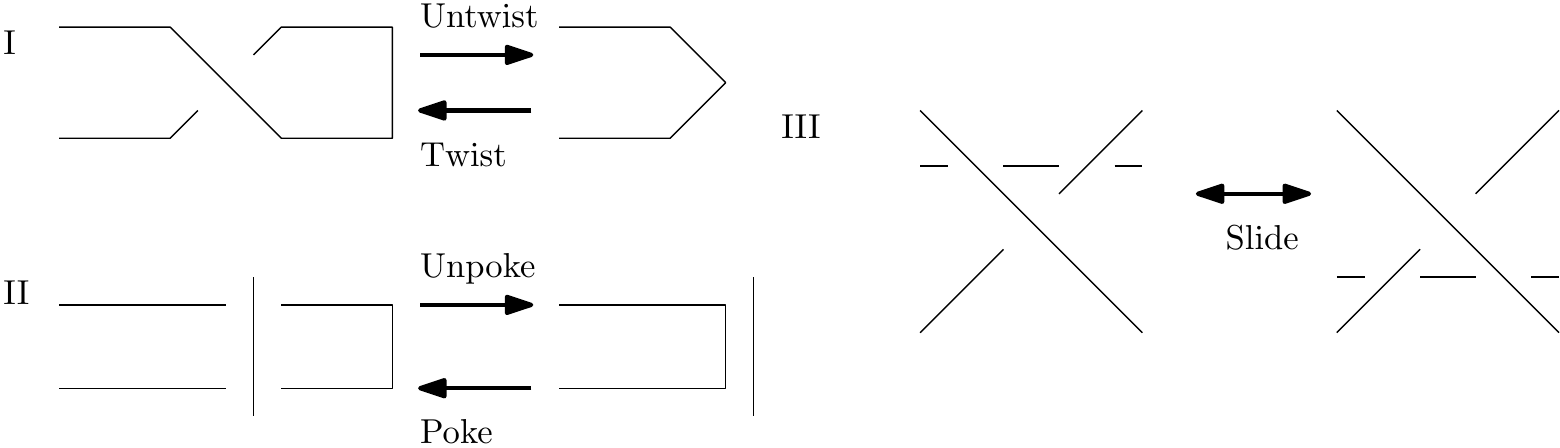}
\end{center}
\label{figure:reidemeister}
\caption{Reidemeister moves}
\end{figure}
 
% Thus the most natural approach to the unknot recognition problem would be to try possible is to try to find an explicit upper bound on the number of Reidemeister moves required to turn a given diagram of the unknot with c crossings into the trivial diagram.
In particular, the diagram of every unknot can be reduced to the trivial diagram (a circle) by performing 
Reidemeister moves.
While each of the Reidemeister moves  can be performed in polynomial time, it is very unclear how many of these moves are required to move an unknot to the trivial diagram.  The problem is that sometimes a successful unknotting sequence of  Reidemeister moves is not monotone, that is, it  has to increase the number of crossings (vertices) in the knot diagram, see e.g.~\cite{Lack15}. Bounding the number of required Reidemeister moves by any function on the number of vertices in the knot diagram was a long-standing open question in the area. The answer to this question was given by Hass and  Lackenby~\cite{HassL01} who gave the first (exponential) upper bound on the number of Reidemeister moves. Later Lackenby in \cite{Lack15} improved the bound significantly by showing that 
any diagram of the unknot with $n$ crossings may be reduced to the trivial diagram using at most 
$(236n)^{11}$ Reidemeister moves. Let us note that this also implies that unknot recognition problem is in NP. 

In this work we consider the unknotting problem when the given knot diagram has treewidth at most $2$. We defer the definition of treewidth till the next section. Our main algorithmic result is the following theorem.
\begin{theorem}\label{thm:mainalgo}
Deciding whether any diagram with $n$ crossings and treewidth at most $2$ is a diagram of the unknot can be decided in time $O(n)$.
\end{theorem} 
 
Our proof yields also the following combinatorial result about the number of Reidemeister moves. It is interestingly to note that in Theorem~\ref{thm:maincomb} we do not use the slide move.
 
\begin{theorem}\label{thm:maincomb}
Any diagram of treewidth $2$  of the unknot with $n$ crossings may be reduced to the trivial diagram using at most $n$ (un)twist and (un)poke Reidemeister moves.
 \end{theorem}

 Actually, the techniques developed to prove Theorems~\ref{thm:mainalgo} and \ref{thm:maincomb} can be used to solve a slightly more general problems about links, that is  collection of knots which do not intersect,  with diagrams of treewidth $2$.
 
 \medskip\noindent
 \textbf{Related work.} To the best of our knowledge, the question whether the unknotting problem with diagrams of bounded treewidth can be resolved in polynomial time is open.  Makowsky and   Mari\~{n}o
in  \cite{MakowskyM03} studied the parametrized complexity of the knot (and link) polynomials known as Jones polynomials, Kauffman polynomials and HOMFLY polynomials on graphs of bounded treewidth. 
 Ru\'{e} et al.~\cite{RueTV18}  
studied the class of link-types that admit a $K_4$-minor-free diagram (which is of treewidth at most $2$).
They  obtain counting formulas and asymptotic estimates for the connected $K_4$-minor-free link and unknot diagrams.

 %https://en.wikipedia.org/wiki/Unknotting_problem
% 
% While there are several unknotting algorithms, see \cite{Haken61,Jaco95}, the major unresolved challenge is whether there existing a polynomial time algorithm solving the  unknotting problem.
% By the celebrated result of Hass,   Lagarias,   and Pippenger~\cite{HassLP99},  unknot  recognition is in NP. 
% \todo[inline]{Agol gave a sketch of the proof that the problem is also in  
%co-NP , (see also an alternative approach   assuming the Generalised Riemann Hypothesis by Kuperberg \cite{Kuperb14}).}
% 
% 
% 
% 
%Introduction to the problem book on knots and links \cite{MR0515288}
%
%\cite{MakowskyM03}
%
%\cite{de2018tree}
%
%Introduction to graphs of treewidth two
%Maybe cite \cite{ArnborgP86}
%
%Introduction to knot theory. Rademeister moves. Equivalence of knot diagrams
%
%Description of our main results...
%
%\todo[inline]{Our algorithm: monotone sequence of 
%Reidemeister moves. Not true in general. 
%}
%\todo[inline]{
%Comparison to Ru\'{e} et al.~\cite{RueTV18}. Research done independently; graph class is the same (a graph has treewidth at most two, if and
%only if it does not contain $K_4$ as a minor); our paper gives a linear time algorithm to classify a diagram, while their paper enumerates
%}

\section{Preliminaries}
\label{section:preliminaries}

\subsection{Treewidth}
We first define the notion of treewidth, as introduced by Robertson and Seymour~\cite{RobertsonS2}. There are many equivalent definitions, see e.g., \cite{Bodlaender98}; the most common definition is the following.

A tree decomposition of a graph $G=(V,E)$ is a pair $(\{X_i~|~i\in I\}, T=(I,F))$, with $\{X_i~|~i\in I\}$ a family of subsets of $V$, and $T$ a tree, such that $\bigcup_{i\in I} X_i=V$, for all $\{v,w\}\in E$, there is an $i\in I$ with $v,w\in X_i$, and for all $v\in V$, the set $\{i\in I~|~v\in X_i\}$ forms a connected subtree of $T$. The {\em width} of a tree decomposition is $\max_{i\in I} |X_i|-1$, and the {\em treewidth} of a graph $G$ is the minimum width of a tree decomposition of $G$.

In this paper, we focus on graphs of treewidth two. We do not need the representation by tree decompositions in this paper, but instead rely on a simple procedure that can be used to recognize graphs of treewidth 2 (Theorem~\ref{theorem:folklorereduce}).
% Also useful is a characterization of treewidth 2 graphs by $K_4$ as forbidden minor (Theorem~\ref{theorem:k4}). %Not needed
Graphs of treewidth two are sometimes also called {\em series-parallel}, but as there are different definitions of what is a series-parallel graph (e.g., whether $K_{1,3}$ is series-parallel depends on the used definition), avoid ambiguity by using treewidth terminology.

The following result is well known.

\begin{theorem}[Folklore, see e.g.,~\cite{Bodlaender98}]
A graph $G=(V,E)$ has treewidth at most two, if it can be reduced to the empty graph by repeating the following operations, while possible:
\begin{itemize}
\item Remove a vertex of degree 0.
\item Remove a vertex of degree 1 and its incident edge.
\item Contract a vertex of degree 2 with a neighbor (possibly creating a parallel edge).
\item Remove one of a pair of parallel edges.
\end{itemize}
\label{theorem:folklorereduce}
\end{theorem}

Notice, the operations from Theorem~\ref{theorem:folklorereduce} are applied to the vertices of degree up to two. Thus, for attacking the unknoting problem on a knot diagram of treewidth two, we can apply the above reduction operations to the small degree vertices of the knot diagram while maintaining ambient isotopy of the knots. This results in small unknots, knots, unlinks and links which are trivially recognizable.

A {\em subdivision} in a graph $G=(V,E)$ is a vertex of degree two. The operation to {\em add a subdivision} is the following: take an edge $\{v,w\}$, and replace this edge by edges $\{v,x\}$ and $\{x,w\}$ with $x$ a new vertex. The operation to {\em remove} a subdivision is the following: take a vertex of degree 2, add an edge between its neighbors and then remove the vertex and its incident edges.

\subsection{Knot diagrams}
In this section, we introduce the notion of a knot (link) diagram. In the literature, variations on this definition are used, where the most common definition is as follows. The knot (link) diagram is an immersed plane curve which is a projection of a knot (link) on the plane with the additional data of which strand/string is over and which is under at each crossing. We slightly extend this definition keeping precisely the same expressive power. Namely, we introduce a plane graph describing the knot diagram. This graph is useful for our algorithm.

Given a projection of a knot on the plane, we put a vertex at each string crossing. Conventionally, we assume that the knot projection is such that every vertex represents a crossing of exactly two strings. Two vertices are adjacent if they represent two consecutive crossings on a string. Thus, all vertices introduced so far have degree four. Notice, two consecutive crossings on a string could be at the same vertex, meaning there is a (self)loop on that vertex. Moreover, two vertices might represent consecutive crossings on more than one string, meaning there are parallel or multi-edges. To make the graph simple, for each loop and each multi-edge we introduce a degree two vertex subdividing the loop/multi-edge. Now, each vertex of a simple graph have either degree two or degree four. We refer to the resulting undirected simple plane graph $G=(V,E)$ as a {\em knot diagram}. 

For keeping the information which string is over and which string is under at each crossing, we label the endpoints of the edges. Consider a vertex of degree four representing a crossing of two strings. Consider an edge representing an overcrossing string at that vertex. We label the endpoint of the overcrossing edge with $u$ (for ``up''). The endpoints of the undercrossing edges are labeled with  $d$ (for ``down''); see Figure~\ref{figure:degree4-normaltype} for an illustration. Now, we label the endpoints at vertices of degree two. If an edge has one (labeled) endpoint at a vertex of degree four and an endpoint at a vertex of degree two, the endpoint at the vertex of degree two receives the same label as the endpoint at the vertex of degree four. Notice, the labels of the same edge at two endpoints might be different. 

\begin{figure}[htb]
\begin{center}
\includegraphics[scale=0.8]{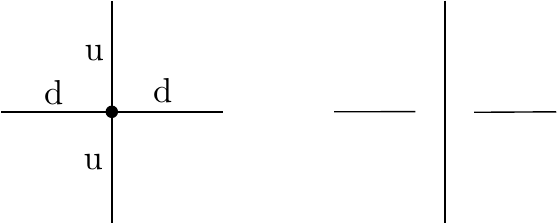}
\end{center}
\caption{A vertex of degree four representing a crossing of two strings}
\label{figure:degree4-normaltype}
\end{figure}

%\begin{theorem}
%{\em Needed: theorem that tells that we can take the sum if there is an edge separator of size 2}
%\label{theorem:edgeseparator}
%\end{theorem}
%
%Note that in a knot diagram, we have two types of vertices:
%\begin{itemize}
%\item Vertices of degree 2
%\item Vertices of degree 4, where the four incident edges are paired, as described above, with one pair labeled up and one
%pair labeled down.
%\end{itemize}

\subsection{Generalized knot diagrams}
Our algorithm is based upon a generalization of knot diagrams, which we call {\em generalized knot diagrams}. The main ingredient is a new type of edges, which are created in the course of the algorithm. While a {\em single} edge in a knot diagram represents a piece of a single string, a {\em double} edge in a generalized knot diagram represents two pieces of strings between two pairs of vertices of degree two. Specifically, consider any two pieces of strings not intersected by any other piece of string. Let the two pieces of strings have the (four) endpoints at vertices of degree two. Moreover, let the strings alternate at every two consecutive crossings with respect to over- and under-crossing, i.e., if string $s$ is over-crossing string $s'$ at a crossing, then at the next (consecutive) crossing $s'$ is over-crossing $s$. In accordance with Reidemeister terminology, we refer to these alternating crossings as {\em twists}. Such two strings with twists between two pairs of vertices of degree two are referred as double edges.

For each double edge we create an integer label that gives the number of twists/crossings in the double edge.  If the two pieces of string do not cross, the label is zero. 
With labelings of endpoints of the strings with $u$ (up) and $d$ (down), we can distinguish between overcrossings and undercrossings; details are given later in this section.
 
See Figure~\ref{figure:doublewithtwists} for an illustration how a double edge represents two pieces of string with three twists.
\begin{figure}[htb]
\begin{center}
\includegraphics[scale=1.0]{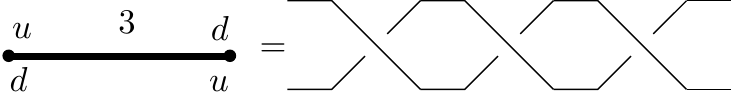}
\end{center}
\caption{A three-twist double edge and its string representation}\label{figure:doublewithtwists}
\end{figure}

In the generalized knot diagram we identify a pair of degree two vertices associated with an endpoint of a double edge as one {\em double vertex}, thus creating a new simple graph with a mix of knot diagram ({\em single}) vertices, double vertices, single and double edges, where double edges are labeled with numbers of twists. 

In the construction and in the algorithm, we will ensure that a double vertex is never incident to three double edges, or to two double edges and one single edge. Thus, a double vertex is either incident to exactly one double edge, or two double edges, or one double and two single edges. The different cases are illustrated in Figure~\ref{figure:doublevertex}.

\begin{figure}[htb]
\begin{center}
\includegraphics[scale=1.0]{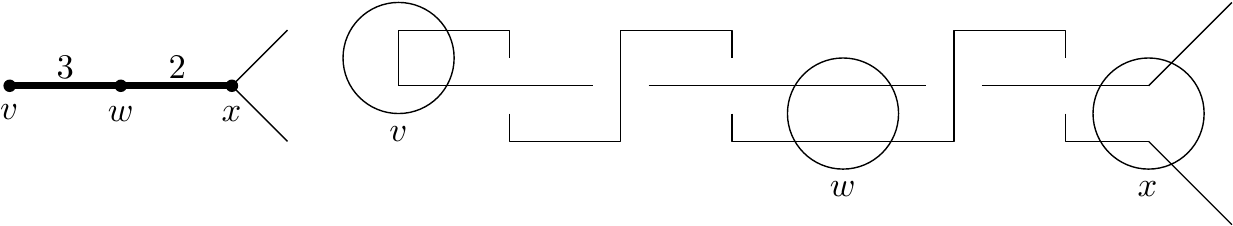}
\end{center}
\caption{The different cases for a double vertex}\label{figure:doublevertex}
\end{figure}

%HB: I do not think this is necessary.
%Notice, generally, double vertices might be incident to three double edges or to two double and one single edge. For instance, if vertices $v$ and $w$ form an endpoint of double edge $e$ while $v$ is also shared by a double edge $e'$ and $w$ is shared by a double edge $e''\neq e'$, the double vertex formed by $v$ and $w$ is incident to three double edges, $e, e', e''$. This creates ambiguity when $v$ and $w$ are treated as a pair and not individually. To avoid the ambiguity and information loss, in the constructions below we make sure that a double vertex is either a pair of adjacent vertices or incident to one double edge and two single edges (one per vertex) or it is an endpoint of exactly two double edges. We refer to this requirement as {\em unambiguity property}. The unambiguity property is initially guaranteed as we start.
 
We also distinguish two types of degrees of vertices in knot diagrams. The {\em graph degree} of a vertex is the number of incident single and double edges. The {\em diagram degree} of a vertex is the number of incident single edges plus twice the number of incident double edges. In a generalized knot diagram, each vertex has diagram degree two or four. Hence, in a generalized knot diagram we have the following types of vertices.
\begin{itemize}
\item A double vertex incident to a double edge. It has graph degree one, and diagram degree two. This is the case when the double vertex is formed by a pair of adjacent vertices. (Vertex $v$ in Figure~\ref{figure:doublevertex}.)
\item A single vertex incident to two single edges. It has graph degree two and diagram degree two.
\item A double vertex incident to two double edges. It has graph degree two and diagram degree four. 
(Vertex $w$ in Figure~\ref{figure:doublevertex}.)
\item A double vertex incident to one double edge and two single edges. It has graph degree three and diagram degree four.
(Vertex $x$ in Figure~\ref{figure:doublevertex}.)
\item A single vertex incident to four single edges. (See Figure
\ref{figure:degree4-normaltype}.)
\end{itemize}

The over- and under-crossing labels at the endpoints of double edges are determined by the $\{u,d\}$-labeling of the endpoints of the two pieces of strings. Consider an endpoint of a double edge at a double vertex formed by two single vertices, $x$ and $y$, both of diagram degree two. Assuming an initial global ordering of single vertices in the knot diagram, without loss of generality, let $x$ precedes $y$ in that order. Then, if the endpoint of a single edge at $x$ has label $u$, we assign to the endpoint of the double edge label $ud$; and $du$ for otherwise. Notice, given a number of twists in a double edge and a $\{u,d\}$-label of any of the single endpoints of the knot diagram, the labels of all other single endpoints and the labels of the double edge endpoints are uniquely determined; see Figure~\ref{figure:doubleedgevertexlabels}.  

\begin{figure}[htb]
\begin{center}
\includegraphics[scale=1.0]{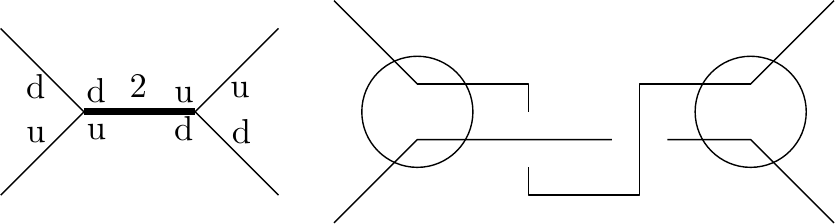}
\end{center}
\caption{Markings determine over- or under-twists}\label{figure:doubleedgevertexlabels}
\end{figure}

When addressing the unknotting (and unlinking) problem, we will apply the reduction operations from Theorem~\ref{theorem:folklorereduce} transforming a knot diagram into a {\em generalized knot diagram}. Notice, a knot diagram is a special type of a generalized knot diagram, namely the one where each edge is a single edge. 

\subsection{Reidemeister moves}
With help of Reidemeister moves, see Figure~\ref{figure:reidemeister}, we obtain an equivalence relation on knot diagrams: if a diagram can be obtained from another by zero or more Reidemeister moves, then these diagrams are equivalent. A diagram is said to be an unknot, if and only if it is equivalent to a diagram with only vertices of diagram degree 2. To simplify notations, we also allow to add subdivisions of edges, or to remove vertices of diagram degree 2 and connect their neighbors.

\section{Main Algorithm}
\label{section:algorithm}
We first give the main shape of the algorithm. 
The algorithm maintains a generalized knot diagram and the underlying simple graph, obtained by ignoring all markings and ignoring the difference between double and single edges. Let us call this underlying
simple graph $G^s$. Note that $G^s$ has treewidth at most 2, and the degree of a vertex in $G^s$ is its graph degree.

With help of the reduction rules of the next section, we build our linear time algorithm. Our algorithm can be seen as a variation of the 
reduction algorithm for treewidth 2 graphs, as described in Theorem~\ref{theorem:folklorereduce}.

The main form of our algorithm is the following. The input is a knot diagram $K$ of treewidth 2.

\begin{itemize}
\item Subdivide parallel edges in $K$, and subdivide selfloops. $K$ now is a generalized knot diagram, yet without double edges, but without parallel edges or selfloops. Compute the underlying simple graph $G^s$.
\item Repeat till we decided that we have a knot, a link, or $G^s$ has at most three vertices:
\begin{itemize}
\item Take a vertex $v$ of graph degree at most 2.
\item Apply a safe rule (defined in the next section) to $K$, that removes $v$. Let again $G^s$ be the resulting simple graph. 
\end{itemize}
\item If $G^s$ has at most three vertices, then classify the generalized knot diagram $K$ with help of a simple case analysis.
\end{itemize}

In the description above, equivalence is topological: if we transform a generalized knot diagram $K_1$ to a new generalized knot diagram $K_2$, $K_1$ represents a knot diagram that can be obtained by Reidemeister moves from a knot diagram that is represented by $K_2$.

We can distinguish three different cases for the vertex $v$ of degree at 
most 2 in the algorithm above: $v$ has graph degree one, and hence is incident to one double edge; $v$ has graph degree 2 and is incident to two single edges; $v$ has graph degree 2 and is incident to two double edges. Each of these cases is discussed in detail in the next section, where we see how we can transform in constant time the generalized knot diagram to an equivalent one, which has the new $G^s$ as underlying simple graph. As in each round, we lose one vertex, the algorithm has at most
$n$ rounds. 

Standard algorithmic techniques allow to implement the selection of $v$ and maintenance of $G^s$ in linear time. We use the adjacency list data structure to represent $G^s$. In addition, we have a set data structure $S$ for vertices of graph degree at most 2. If we delete $v$,
we check if the neighbors of $v$ have their degree decreased to a value at most 2, and if so, add these to $S$. Selecting $v$ can be done by just taking an element from $S$. We can e.g. use a standard queue for $S$.

What remains is to look at the rules that deal with each of the cases, which is done in the next section.

%\begin{theorem}
%Suppose we have a non-empty generalized knot diagram of treewidth two, with at least three vertices. Then, either one of the rules applies, or we have a graph with three vertices.
%\end{theorem}
%
%
%\begin{theorem}
%All rules are safe.
%\end{theorem}
%
%\begin{theorem}
%Applying rules while possible can be done in $O(n)$ time for treewidth two graphs.
%\end{theorem}

\section{Safe Reduction Rules}
\label{section:rules}

In this section we introduce a number of reduction rules for generalized knot diagrams. The result of a rule is always again a generalized knot
diagram. Note that we always remove one (single or double) vertex of graph degree at most 2, and possibly add an edge between the neighbors of a removed vertex of degree 2. Thus, when any of these rules is applied, the size of the generalized knot diagram is decreased by at least one vertex. 

A rule is \emph{safe} if application of the rule preserves the ambient isotopy of the original and the resulting knots. For all rules in this section, the knots equivalence is maintained due to the fact that each of the rules is either a Reidemeister move(s) or removing the subdivision. 

%The unambiguity property is maintained by construction: we assume that a given generalized knot diagram satisfies the property and we prove that the reduction rules do not create double vertices of graph degree three. We refer to the rules satisfying these two properties as {\em safe} rules.   

%We apply the rules while the graph has at least three vertices.
%As a consequence, all vertices will have a 1-degree that is 1, 2, 3 or 4. Theorem~\ref{theorem:folklorereduce} gives that there is always a vertex of 1-degree 1 or 2. 

%We say that a rule is {\em safe}, if whenever we obtain a generalized knot diagram $G'$ by applying the rule to generalized knot diagram $G$, the knot diagram $\overline{G'}$ can be obtained from the knot diagram $\overline{G}$ by Reidemeister moves (and thus, these represent equivalent knots), and removing or adding subdivisions. 

\subsection{Vertices incident to one double edge}
We first look at vertices of graph degree 1. Such a vertex must be a double vertex, incident to a double edge. We have the following rule.
%Figure~\ref{figure:rule1a} illustrates the rule, for the case when $w$ is incident to two single edges.
%We first consider a double vertex of graph degree one incident to a double edge. 

\begin{myrule}
\label{rule:degree1A}
Let $\{v,w\}$ be a double edge, with the graph degree of $v$ equal to one. Remove $v$ and the edge $\{v,w\}$. If double vertex $w$ was incident to two single edges $e$ and $e'$, make $w$ a single vertex of degree two incident to the same edges and keep the $\{u,d\}$-labels at the endpoint $w$ of $e$ and $e'$ intact; this case is illustrated in Figure~\ref{figure:rule1a}. If $w$ was incident to another double edge, keep $w$ a double vertex formed by two adjacent single vertices.
\end{myrule}

\begin{figure}[htb]
\begin{center}
\includegraphics[scale=0.8]{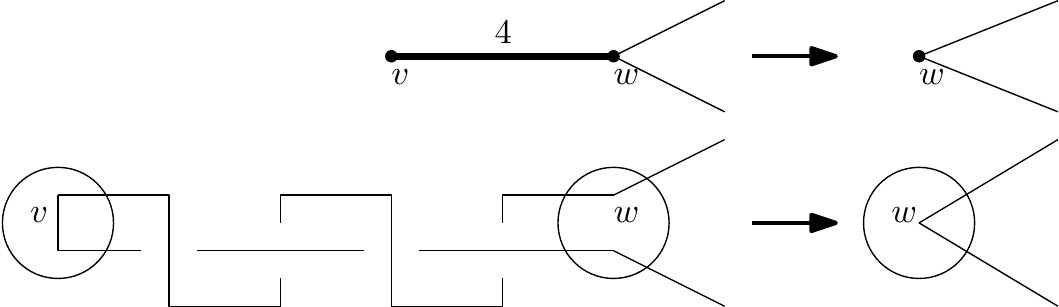}
\end{center}
\caption{Removing a double vertex of graph degree one.}
\label{figure:rule1a}
\end{figure}

\begin{lemma}
Rule~\ref{rule:degree1A} is safe.
\end{lemma}
\begin{proof}
Let $\{v,w\}$ be a double edge having $i$ twists. We first apply $i$ Reidemeister untwists to the double edge $\{v,w\}$, resulting in a subdivided path between two single vertices forming the double vertex at $w$. Then, the rule introduces two cases. If $w$ is incident to two single edges, we contract the entire path into a single vertex incident to two single edges. If $w$ is incident to a double edge, we contract the path to an edge between two vertices forming a double vertex. In both cases we execute Reidemeister moves followed by contraction of a subdivided path. As both operations are safe, the rule is safe as well.  
\end{proof}

\subsection{Vertices incident to two single edges}
Consider a single vertex $v$ incident to two single edges $\{v,w\}$ and $\{v,x\}$. We consider three cases: 1) $w$ and $x$ are not adjacent; 2) $w$ and $x$ are adjacent by a single edge; 3) $w$ and $x$ are adjacent by a double edge. Consider the first case.

\begin{myrule}
\label{rule:2simplea}
Let $v$ be incident to two single edges $\{v,w\}$ and $\{v,x\}$, where $w$ and $x$ are not adjacent. Remove $v$ and the edges $\{v,w\}$ and $\{v,x\}$, and add a single edge $\{w,x\}$. Keep the $\{u,d\}$-labels of edge $\{w,x\}$ at $w$ and $x$ the same as in $\{v,w\}$ and $\{v,x\}$, respectively.
\end{myrule}

Rule~\ref{rule:2simplea} is trivially safe as we just remove an edge subdivision. See Figure~\ref{figure:rule2-noedge}.

\begin{figure}[htb]
\begin{center}
\includegraphics[scale=1.0]{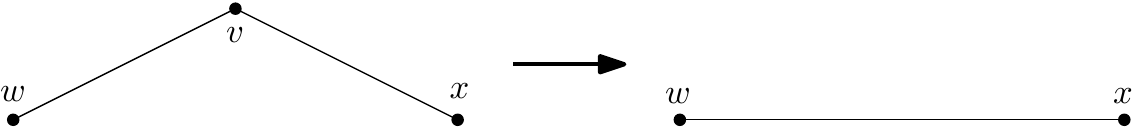}
\end{center}
\caption{Removing a vertex with two single edges and nonadjacent neighbors.}
\label{figure:rule2-noedge}
\end{figure}

Consider the second case where $w$ and $x$ are adjacent by a single edge. 

\begin{myrule}
Let $v$ be incident to two single edges $\{v,w\}$ and $\{v,x\}$, where $w$ and $x$ are adjacent by a single edge. 
\label{rule:degree2ss}
\begin{enumerate} 
%Case 1
\item If all three vertices $v,w,x$ are single vertices of graph degree 2, forming a simple cycle of length three, we do not apply any further rules as the diagram represents the unknot.
%Case 2
\item If $w$ is a single vertex of graph degree 2 and $x$ is a single vertex of graph degree 4, we delete vertices $v$ and $w$ together with edges $\{v,w\}, \{v,x\}$ and $\{w,x\}$. 
%Case 3
\item If $w$ is a single vertex of graph degree 2 and $x$ is a double vertex of graph degree 3, we delete vertices $v$ and $w$ together with edges $\{v,w\}, \{v,x\}$ and $\{w,x\}$, and we make the two vertices forming double vertex $x$ adjacent. 
\begin{figure}[htb]
\begin{center}
\includegraphics[scale=0.85]{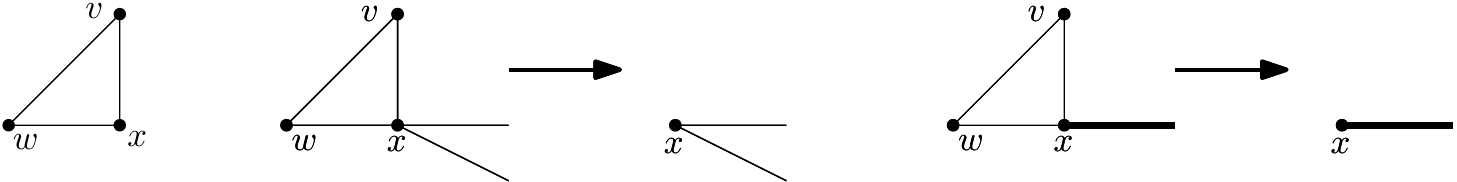}
\end{center}
\caption{The first three cases of Rule~\ref{rule:degree2ss}.}
\label{figure:r2-cases12}
\end{figure}
%Case 4
\item If $w$ and $x$ are two single vertices of graph degree 4 and the $\{u,d\}$-labels of single edges $\{v,w\}$ and $\{v,x\}$ at vertices $w$ and $x$ (respectively) are the same, we delete vertex $v$ together with edges $\{v,w\}, \{v,x\}$ and $\{w,x\}$, and we create a double edge $\{w,x\}$ of 0 twists. To assign the $\{u,d\}$-labels at vertices $w$ and $x$, we split each of these vertices into two non-adjacent copies, one per single edge incident to the vertex. We connect the over-/undercrossing endpoint of $\{w,x\}$ at $w$ with the over-/undercrossing endpoint of $\{w,x\}$ at $x$, respectively, see Figure~\ref{figure:r2-case4}.
% We assign $u$-labels to the overcrossing endpoints, and $d$-labels to the undercrossing ones. 
%\footnote{Hans: I don't think labels are needed. There are no crossings here; see Figure 9.}
\begin{figure}[htb]
\begin{center}
\includegraphics[scale=1.0]{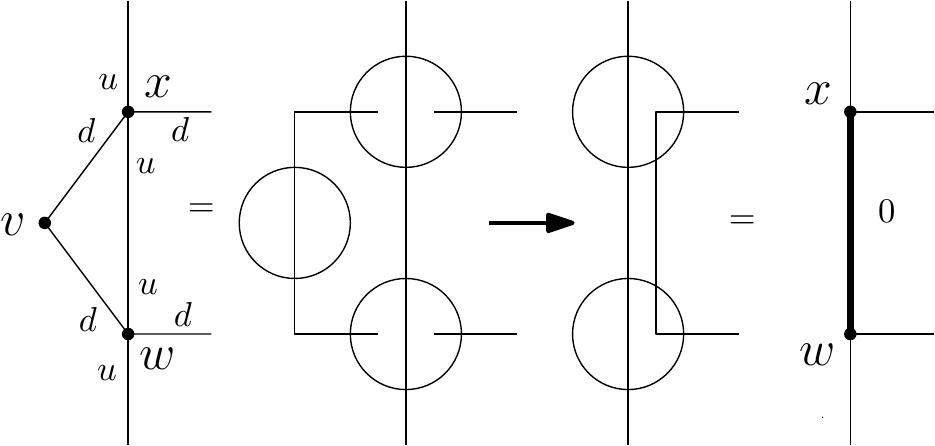}
\end{center}
\caption{The fourth case of Rule~\ref{rule:degree2ss}.}
\label{figure:r2-case4}
\end{figure}
%Case 5
\item If $w$ and $x$ are two single vertices of graph degree 4 and the $\{u,d\}$-labels of single edges $\{v,w\}$ and $\{v,x\}$ at vertices $w$ and $x$ (respectively) are different, we delete vertex $v$ together with edges $\{v,w\}, \{v,x\}$ and $\{w,x\}$, and we create a double edge $\{w,x\}$ of 2 twists. To assign the $\{u,d\}$-labels at vertices $w$ and $x$, we again split each of these vertices into two non-adjacent copies, one per single edge incident to the vertex. We keep the $\{u,d\}$-labels at endpoints of $\{w,x\}$ identical to the labels of the incident single edges, see Figure~\ref{figure:r2-case5}.
\begin{figure}[htb]
\begin{center}
\includegraphics[scale=1.0]{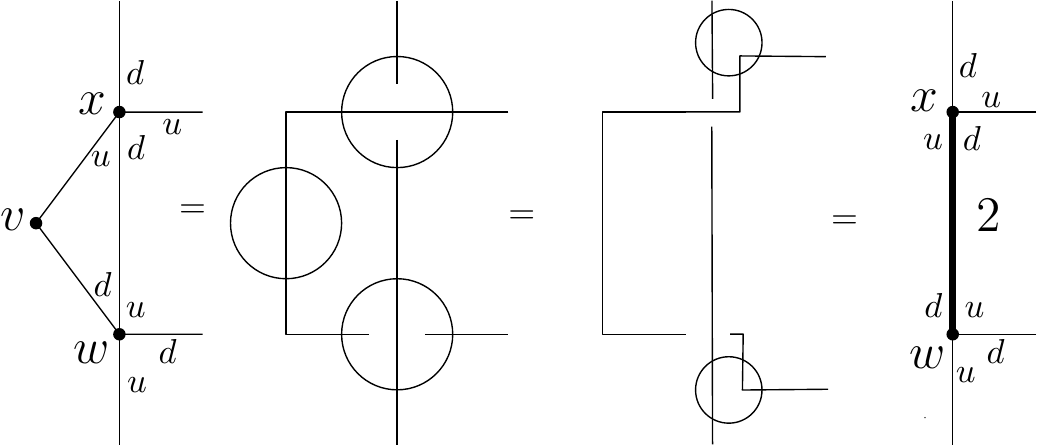}
\end{center}
\caption{The fifth case of Rule~\ref{rule:degree2ss}.}
\label{figure:r2-case5}
\end{figure}
%Case 6
\item If $w$ is a single vertex of graph degree 4 and $x$ is a double vertex of graph degree 3, we delete vertex $v$ together with edges $\{v,w\}, \{v,x\}$ and $\{w,x\}$, we split vertex $w$ into two non-adjacent copies (one per string at the crossing represented by vertex $w$), thus creating a double vertex $w$, and we create a double edge $\{w,x\}$ of 1 twist. Now, we assign the $\{u,d\}$-labels at the endpoints of the new double edge $\{w,x\}$. At double vertex $w$ the endpoint corresponding to the overcrossing (undercrossing) string receives label $u$ (and $d$), respectively. Since the number of twists in double edge $\{w,x\}$ is 1, at double vertex $x$ the labels at the endpoints of the strings alternate, see Figure~\ref{figure:r2-case6}.
\begin{figure}[htb]
\begin{center}
\includegraphics[scale=1.0]{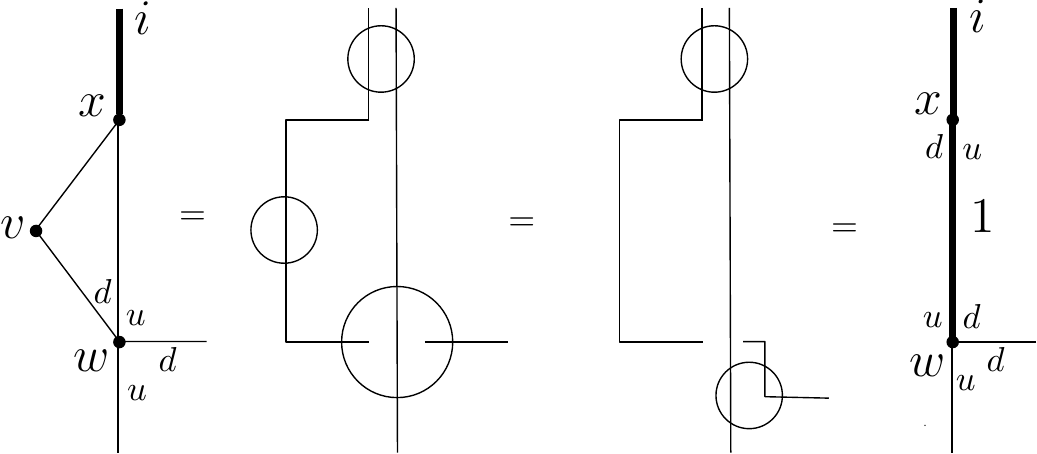}
\end{center}
\caption{The sixth case of Rule~\ref{rule:degree2ss}.}
\label{figure:r2-case6}
\end{figure}
%Case 7
\item If both $x$ and $w$ are double vertices of graph degree 3, we delete vertex $v$ together with edges $\{v,w\}, \{v,x\}$ and $\{w,x\}$, and we create a double edge $\{w,x\}$ of 0 twists, keeping the $\{u,d\}$-labels at endpoints as is.
\end{enumerate}
\end{myrule}

\begin{lemma}
Rule~\ref{rule:degree2ss} is safe.
\end{lemma}
\begin{proof}
The first three cases are illustrated in
Figure~\ref{figure:r2-cases12}. The first case is the trivial unknot. In the second case, we do one Reidemeister untwist that removes the crossing at $x$, and subsequently we remove the subdivision. The third case is just removing the subdivision.

The fourth case, illustrated in Figure~\ref{figure:r2-case4} is a Reidemeister unpoke and removing the subdivision. In this case we see we still have a generalized knot diagram 
because each of the newly created double vertices is incident to one double and two single edges by construction.

In the fifth and in the sixth cases, see Figure~\ref{figure:r2-case5} and Figure~\ref{figure:r2-case6}, we keep exactly the same
knot diagram, but the generalized knot diagram representing the knot is changed to one with fewer vertices. The same holds for the seventh case.
\end{proof}

%The fifth, sixth and seventh cases are all done by subdivided edge contraction and construction of a double edge by definition. The unambiguity is maintained again by construction.

\medskip

We now look at the third main case. Suppose $v$ is adjacent by two single edges to
two double vertices, $w$ and $x$, and there is a double edge between $w$ and $x$
having $i$ twists.

\begin{myrule}
\label{rule:4}
Suppose $v$ has two single edges, to
$w$ and $x$ respectively, and there is a double edge between $w$ and $x$
with $i$ twists.
\begin{enumerate}
\item If $i =0$, then the generalized knot diagram represents an unlink. We recurse on the part of the graph having more than three vertices.
\item If $i\neq 1$ is odd, the generalized knot diagram represents a knot;
\item If $i\neq 0$ is even, the generalized knot diagram represents a link.
\item If $i = 1$, then delete vertex $v$ together with adjacent edges and delete double edge $\{w,x\}$, make $w$ and $x$ single vertices adjacent by a single edge.
\end{enumerate}
\label{rule:degree2sd}
\end{myrule}

%Consider the third main case where $w$ and $x$ are adjacent by a double edge of $i$ twists. In this case for three sub-cases no further rules are needed as these sub-cases are completely and uniquely classifiable\footnote{Alex: figures or the text is enough?}:
%\begin{itemize}
%\item If $i=0$, the generalized knot diagram represents an unlink;
%\item If $i\neq 1$ is odd, the generalized knot diagram represents a knot;
%\item If $i\neq 0$ is even, the generalized knot diagram represents a link.
%\end{itemize}

Correctness of the first three cases  is evident; see Figure~\ref{figure:rule4-case123}. 
Safeness of the fourth case follows as this step
represents a single Reidemeister untwist with subsequent contraction of subdivision, see Figure~\ref{figure:rule4-case4}. 

\begin{figure}[htb]
\begin{center}
\includegraphics[scale=1.0]{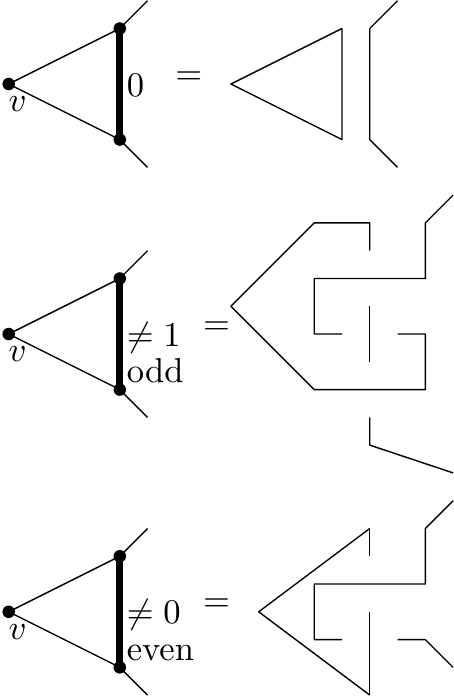}
\end{center}
\caption{The first three cases of Rule~\ref{rule:degree2sd}.}
\label{figure:rule4-case123}
\end{figure}

\begin{figure}[htb]
\begin{center}
\includegraphics[scale=1.0]{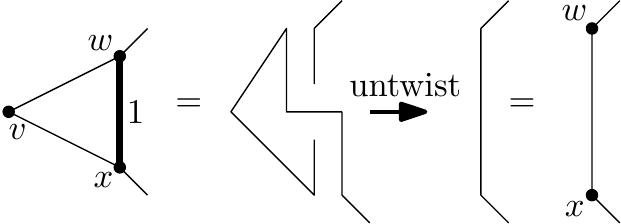}
\end{center}
\caption{The fourth case of Rule~\ref{rule:degree2sd}.}
\label{figure:rule4-case4}
\end{figure}
 
\subsection{Vertices incident to two double edges}  
Consider a double vertex $v$ incident to two double edges $\{v,w\}$ and $\{v,x\}$ of $i$ and $j$ twists, respectively. We consider three cases: 1) $w$ and $x$ are not adjacent; 2) $w$ and $x$ are adjacent by a single edge; 3) $w$ and $x$ are adjacent by a double edge. Consider the first case.

\begin{myrule}
\label{rule:5}
Let $v$ be incident to two double edges $\{v,w\}$ and $\{v,x\}$, where $w$ and $x$ are not adjacent. Remove $v$ and the edges $\{v,w\}$ and $\{v,x\}$, and add a double edge $\{w,x\}$. If the endpoints of $\{w,v\}$ and $\{v,x\}$ at $v$ were agreeing on the $\{u,d\}$-labels, i.e., the labels of the endpoints on the same strings were the same, define the number of twists on the new double edge $\{w,x\}$ by $|i-j|$. If the number of twists $i$ in double edge $\{v,w\}$ is greater than the number of twists $j$ in double edge $\{v,x\}$, keep the $\{u,d\}$-labels at $w$ intact and alternate at $x$, otherwise keep the labels at $x$ and alternate at $w$. 
\begin{figure}[htb]
\begin{center}
\includegraphics[scale=1.0]{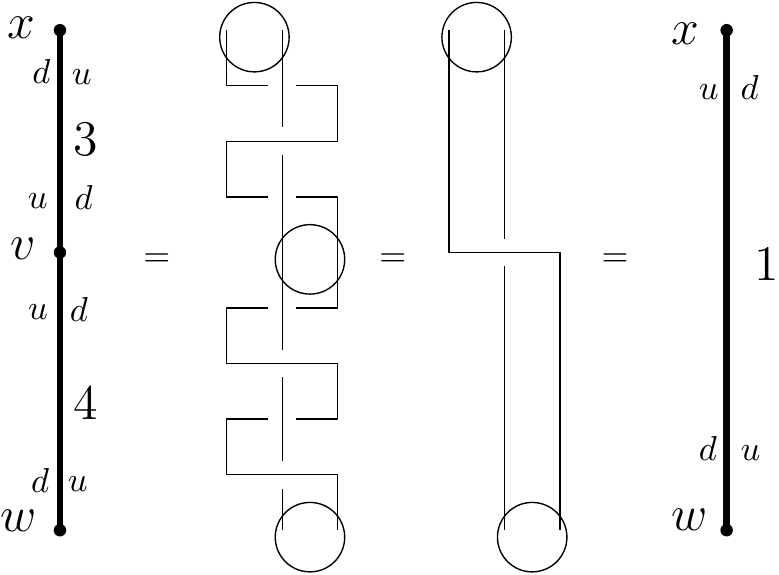}
\end{center}
\caption{Rule~\ref{rule:5} with agreeing labels at $v$.}
\label{figure:rule5-agree}
\end{figure}

If the endpoints were disagreeing on the labels, i.e., the labels of the endpoints on the same strings were different, define the number of twists of edge $\{w,x\}$ by $i+j$. Keep the $\{u,d\}$-labels of edge $\{w,x\}$ at $w$ and $x$ the same as in $\{v,w\}$ and $\{v,x\}$, respectively.
\begin{figure}[htb]
\begin{center}
\includegraphics[scale=1.0]{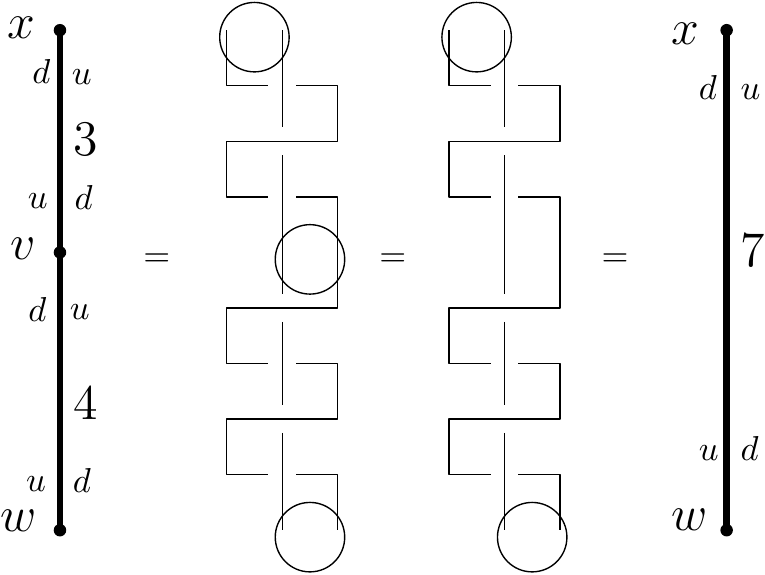}
\end{center}
\caption{Rule~\ref{rule:5} with disagreeing labels at $v$.}
\label{figure:rule5-disagree}
\end{figure}
\end{myrule}

\begin{lemma}
Rule~\ref{rule:5} is safe.
\end{lemma}
\begin{proof}
In the case of agreement on labels, see Figure~\ref{figure:rule5-agree}, we proceed with $\min\{i,j\}$ Reidemeister unpoke moves followed by removing the subdivision. In case of label disagreement, see Figure~\ref{figure:rule5-disagree}, we keep exactly the same knot diagram, but simplify the generalized knot diagram by removing the subdivision on the double edge. In the latter case the number of twists on the new double edge is exactly the sum $i+j$ of the twists on the two original double edges. 
\end{proof}

\begin{myrule}
\label{rule:3reduced}
Let $v$ be incident to two double edges, $\{v,w\}$ and $\{v,x\}$, where $w$ and $x$ are adjacent by a single edge. This case can be transformed to the case considered in Rule~\ref{rule:4}. First, we add a subdivision of $\{w,x\}$. Then, we apply Rule~\ref{rule:5} arriving in and applying one of the cases of Rule~\ref{rule:4}, see Figure~\ref{figure:avoidcase2double}. Therefore, we either classify the diagram or safely reduce the graph.
\begin{figure}[htb]
\begin{center}
\includegraphics[scale=1.0]{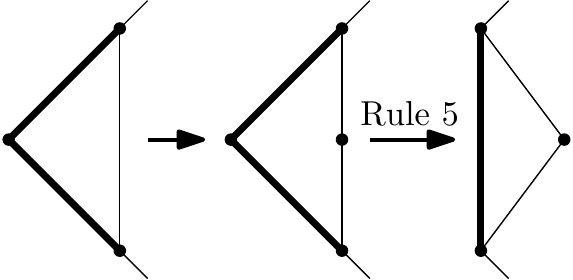}
\end{center}
\caption{Rule~\ref{rule:3reduced} is reduced to rule Rule~\ref{rule:4}}
\label{figure:avoidcase2double}
\end{figure} 
\end{myrule}

%In the second case, when $w$ and $x$ are adjacent by a single edge, $w$ and $x$ have graph degree three and each of them is incident to one double and two single edges. Most of the sub-cases of this case are completely and uniquely classifiable and no rules are further needed. To classify these sub-cases, we have to compute the sum of the twists on both double edges. If the endpoints of $\{w,v\}$ and $\{v,x\}$ at $v$ are agreeing on the $\{u,d\}$-labels, i.e., the labels of the endpoints on the same strings are the same, define $\ell=|i-j|$, otherwise let $\ell=i+j$. The classifiable sub-cases are 
%\begin{itemize}
%\item If $\ell=0$, the diagram represents an unlink;
%\item If $\ell\neq 1$ is odd, the diagram represents a knot;
%\item If $\ell\neq 0$ is even, the diagram represents a link.
%\end{itemize}  

\begin{myrule}
\label{rule:7}
Consider three double vertices $v, w$ and $x$ with three double edges, $\{v,w\}, \{x,v\}$ and $\{w,x\}$ of $i, j$ and $k$ twists, respectively. In this case, the knot diagram is completely and uniquely classifiable. We apply Rule~\ref{rule:5} to  double vertex $v$ obtaining a generalized knot diagram with only two double vertices $w$ and $x$ adjacent by two parallel double edges. The original edge $\{w,x\}$ has $k$ twists. The new edge $\{w,x\}$ has $\ell$ twists, where $\ell$ is either $|i-j|$ or $i+j$ dependent on the case for Rule~\ref{rule:5}, i.e., on the agreement of $\{u,d\}$-labels at $v$. Next, we apply Rule~\ref{rule:5} again to double vertex $w$ obtaining a generalized knot diagram with a single double vertex $x$ and a self-loop with $m$ twists, where $m$ is either $|\ell-k|$ or $\ell+k$ dependent on the agreement of $\{u,d\}$-labels at $w$. Finally, we are ready to classify the knot diagram:
\begin{enumerate}
\item If $m=0$, then the generalized knot diagram represents the unlink;
\item If $m\neq 1$ is odd, the generalized knot diagram represents the $(m,2)$-torus knot, for definition and notations see~\cite{wolfram-torus-knot};
\item If $m\neq 0$ is even, the generalized knot diagram represents the $(m,2)$-torus link;
\item If $m = 1$, then the generalized knot diagram represents the unknot as a single Reidemeister untwist turns the diagram to a circle.
\end{enumerate}
\end{myrule}

\section{Conclusions}
We conclude with the following questions. 
\begin{itemize}
\item 
We gave a linear time algorithm deciding whether any diagram with $n$ crossings and treewidth at most $2$ is a diagram of the unknot.  Is it possible to extend our result to graphs of treewidth $t\geq 3$? Even the existence of a polynomial time algorithm for $t=3$ is open.  
\item 
We also proved that any diagram of treewidth $2$  of the unknot with $n$ crossings may be reduced to the trivial diagram using at most $n$  Reidemeister moves. Is it true, that any diagram of treewidth $t\geq 3$  of the unknot could be reduced to the trivial diagram
in at most $f(t)\cdot n$ moves for some function $f$ of $t$ only?
\item Koenig  and Tsvietkova \cite{koenig2018np}  and 
de Mesmay et al. 
\cite{de2018unbearable}  
proved that deciding if a diagram of the unknot can be untangled using at most $k$ Riedemeister moves (where $k$ is part of the input) is NP-hard. Could this problem be solved in polynomial time on knots with diagrams  of treewidth $2$?

\end{itemize}
%\bibliographystyle{abbrv}
%\bibliography{definitions,papers}

\end{document}